\documentclass[a4paper, 11pt, leqno, twoside]{amsart}
\usepackage{hyperref}

\usepackage{float}

\usepackage{amsmath,amssymb,latexsym,amsthm,enumerate,epsfig,graphicx}

\newtheorem{proposition}{Proposition}[section]

\newtheorem{remark}{Remark}[section]
\numberwithin{equation}{section}

\newcommand{\gcal}{{\mathcal G}}

\newcommand{\jcal}{{\mathcal J}}

\newcommand{\Jcal}{{\mathcal J}}
\newcommand{\const}{{\textsc c}}

\newif\ifcomment \commentfalse
\def\commentON{\commenttrue}

\long\outer\def\BC#1\\EC{\ifcomment \sloppy \par \# \ldots\dotfill
{\em #1} \dotfill \# \par \fi } \commentON

\newcommand{\remove}[1]{}

\newcommand{\Frac}{\displaystyle\frac}

\newcommand{\Dx}{\Delta x}
\newcommand{\Dt}{\Delta t}

\newenvironment{vardesc]}[1]{%
\settowidth{\parindent}{#1: \ }
\makebox{#1:}}{}

\newcommand{\R}{\mathbb{R}}

\newcommand{\s}{\mathop{\textsc{s}}}

\begin{document} 
\title[A non-local rare mutations model in evolutionary dynamics]{A non-local rare mutations model for quasispecies and Prisoner's dilemma: numerical assessment of qualitative behaviour}
\author[A.L.~Amadori, M.~Briani,  R.~Natalini]{Anna Lisa Amadori$^{1}$, Maya Briani$^{2}$, Roberto Natalini$^{2}$}
\address{$^1$Dipartimento di Scienze Applicate, Universit\`a di Napoli \lq\lq Parthenope''}
\address{$^2$ Istituto per le Applicazioni del Calcolo \lq\lq M. Picone'', Consiglio Nazionale delle Ricerche}

\maketitle

\begin{abstract}
An integro-differential model for evolutionary dynamics with mutations is investigated by improving the understanding of its behavior using numerical simulations. The proposed numerical approach can handle also density dependent fitness, and gives new insights about the role of mutation in the preservation of cooperation.
\end{abstract}

\noindent
\subparagraph{{\bf \textup{2010} Mathematics Subject Classification:}} Primary 92D25; Secondary: 35R09, 65M06

\section{Introduction}

Evolutionary dynamics describes the dynamics of populations as the result of  the interplay between ecological interactions  and phenotypic variation.
The ecological mechanism of  replication/selection seems to be well described by an ordinary differential equation, where the rate of growth of any specie (i.e.~the balance between births and deaths) depends on the composition of the entire population.
Such rate, also known as the relative fitness, is actually the difference between the absolute fitness of the species of interest and the mean fitness of the population. 
Several shapes have been proposed for the absolute fitness. When one is modeling phenotypes, the choice of a constant fitness seems fair, but, starting with the seminal work by Maynard Smith and Price \cite{MSP73},  an important amount  of research  deals with ideas arising from mathematical game theory, see \cite{SHbook} and references therein. In this framework, the Prisoner's Dilemma has attracted a lot of attention: we mention \cite{Nowak20121} for a detailed account of the state of the art about this topic.
The first attempt of giving account of mutations, dating back to  the '70, is the so called ``quasispecies equation'', where the growth rate of any specie is modified by considering the dispersion  due to the birth of  mutated offspring. See, for instance, the reference book \cite{ESbook79}. The same underlying idea has been included in the evolutionary games setting with the ``replicator-mutator'' equation in \cite{SS92mutation}. 
An interesting and exhaustive account of evolutionary dynamics can be found in the book  \cite{Nowakbook}.
More recently,  macroscopic PDE models have been proposed and studied (see, for instance, \cite{DJMP05}).
A different approach focus on the stochastic dynamics of each individual in the population, as for instance Dieckmann and Law \cite{DU00}, who analyzed the related moment equations. 
Finally, the emerging field of adaptive dynamics has emphasized that selection and mutation act in two different time-scales, and therefore proposed models based on the combination of jump processes and ordinary differential equations. See, for instance, the trait substitution
sequences \cite{M96adaptive}.
A unifying treatment is provided in \cite{CFM08}, where various macroscopic models  are obtained as the limit of one microscopic model by performing different types of rescaling.

In \cite{ACNT} it has been introduced a macroscopic stochastic model for the selection-mutation process, that goes into the direction of adaptive dynamics. Selection is described by a deterministic differential equation, where the relative fitness rules the reproductive rate.  Mutation, instead, is described by a marked point process. Unlike trait substitution sequence, no assumption is made about the selection dynamics. Specifically it is not asked that there is invasion or extinction of the mutant trait between subsequent mutations. 
The related stochastic differential equation has been studied and a Kolmogorov equation has been deduced and investigated analytically.
Such Kolmogorov equation is of integro-differential type: the non-local term is the deterministic counterpart of the marked point process modeling mutation and therefore shows up even if the total number of strategies is finite.
We focus here on the case of only two different strategies, that can be described by a scalar equation, namely
\begin{align} \label{Ktwo}
\left\{\begin{array}{rlr}
 \partial_tu = & - \s x (1\!-\!x) \partial_xu   - \lambda_1 \gamma_1 f_1 x \,\Jcal\!\left(u,-\gamma_1 x\right)& \\[.15cm]
& +\lambda_{0}\gamma_0 f_0 (1\!-\!x) \, \Jcal\!\left(u, \gamma_0(1-x)\right) , & \quad x\in[0,1], \, t>0 ,\\[.15cm]
u(x,0)=& x, & x\in[0,1],\, t=0. \end{array}\right.
\end{align}
Here $x$ is the (initial) frequency of the type labeled 1, and $u(x,t)$ is the expected frequency of the same type after time $t$.
The functions $f_0(x)$ and $f_1(x)$ stand for the fitness of types $0$ and $1$, respectively, and
\begin{align*}\label{s}
\s(x)= f_0(x)-f_1(x) 
\end{align*}
is the selection spread among species 0 and 1. If $\s>0$ the specie 0 has a selection advantage, and viceversa.
The parameters $\gamma_i$ and $\lambda_i$ are related to the mutation process:  $\gamma_0\in[0,1]$ stands for the proportion of the offspring of individuals of type $0$ that  show a type $1$  by effect of mutation, and $\lambda_0>0$ is related to the time intensity of the point process driving mutations from type $0$ to type $1$, which is given by the product $\lambda_0 f_0(x)$;
$\gamma_1$ and $\lambda_1$ play the same role with respect to mutations from type $1$ to type $0$.
The quantity $\Jcal(u,z)$ is the finite increment related to the point process, precisely
\begin{equation}\label{j0}
\Jcal (u,z) (x,t)= \left\{ \begin{array}{cr}
\left[ u(x+z,t)-u(x,t)\right]/z & \quad \mbox{ if } z\neq 0, 
\\[.15cm] 
0 & \quad \mbox{ if } z =0. \end{array}\right.
\end{equation}
In \cite{ACNT}, the analytical theory of the  solutions  to problem \eqref{Ktwo} was started. Global existence and regularity results were proved together with some results about the qualitative behavior of solutions for long times in the quasispecies case. Even if the main picture was sufficiently clear, which is that there are situations where mutations are able to contrast dominant strategies, it is difficult to say what happens in the general case of density-dependent fitness, for instance in the Prisoner's dilemma case. This is the main motivation to study the problem by a numerical simulation side. 
From the numerical point of view, there are two main difficulties in the approximation of problem \eqref{Ktwo} that require some effort. We shall see in fact that, as time increases, a standard upwind approximation mainly fails due to the blowing up of the first derivative 
$\partial_x u(1,t)$ and also due to the presence of the non-local term.
The numerical investigation then require extremely fine meshes over a small portion of the domain to resolve the solution,  especially for large time simulations.
To deal with these difficulty we shall propose to adopt an adaptive  numerical grid which thickens with the increasing of time near the right boundary $x=1$, and an ad-hoc approximation of the nonlocal term near at $x=1$.

The paper is organized as follows: In Section \ref{2} we settle the model and briefly recall the main analytical results obtained in \cite{ACNT}. We also give some extensions that apply to the density-dependent framework. Section \ref{3} is devoted to the presentation of a suitable numerical scheme and  to some simulations. Section \ref{4} contains some conclusions. 

\section{Analytical framework}\label{2}

Our main concern is exploring the relation among  point type rare mutations, described by equation \eqref{Ktwo}, and continuous rare mutations, described by the canonical replicator mutator equation, see  \cite{SS92mutation},
\begin{equation}\label{r-m}
\dot x = -\s  x(1- x) + m_0 f_0(1- x) - m_1 f_1  x .
\end{equation}
Here $\s$ is the selection spread as before, $m_0$  stands for the mutation probability from species 0 to species 1, and viceversa for $m_1$. 
To  this aim, it is convenient to introduce the flux associated to \eqref{r-m}, i.e.~the solution to the homogeneous transport equation
\begin{equation}\label{qs1} 
\left\{\begin{array}{ll}
 \partial_tv = \left(-\s x(1\!-\!x)  -m_1 f_1 x+ m_0 f_0(1\!-\!x)\right) \partial_x v, \quad & x\in[0,1] , \, t>0, \\[.15cm]
v(x,0)=x , &  x\in[0,1] , \, t=0.\end{array}\right.
\end{equation}
When both $m_0,m_1=0$, \eqref{qs1} gives back
\begin{align}\label{r1} \left\{\begin{array}{ll} \partial_tw =- \s x(1\!-\!x) \partial_x w , \quad &  x\in[0,1] , \, t>0, \\[.15cm] w(x,0)=x , & x\in[0,1] , \, t=0,\end{array}\right. \end{align}
which is  the flux of the simple replicator equation:
\begin{equation}\label{r} \dot x = -\s  x(1- x) .\end{equation}
The analogy between equation \eqref{Ktwo} and \eqref{qs1} suggests  to take 
\begin{align*}\label{lgm}
 m_0=\lambda_{0}\gamma_0, \qquad m_1=\lambda_1\gamma_1. 
\end{align*}
This relation put in evidence that there are many point-mutations models \eqref{Ktwo} related to the same selection-mutation equation \eqref{qs1}.
Indeed, there is a two-dimensional set of parameters $(\gamma_0, \lambda_{0}, \gamma_1, \lambda_1)$  that give back the same $m_0$ and $m_1$. 
When $m_i=0$, we may assume without loss of generality that $\gamma_i=0$. Otherwise, we choose to use $\gamma_i$ as a free parameter, and select $\lambda_{i}=m_i/\gamma_{i}\in[m_i,\infty)$.
As $\gamma_{i}$ goes to $0$, the time intensity  $\lambda_{i}$ increases, and the paths of the point process driving mutations becomes continuous. Similarly, the discrete increment $\Jcal$ approaches the actual derivative $\partial_x$.
On the contrary, at $\gamma_{i}=1$ the time intensity gets its minimum $\lambda_{i}=m_i$, and mutations are concentrated in rare events that happen simultaneously to all the offspring.

In the remainder of this paper  we take $m_0$ and $m_1$ as fixed, and write $v(x,t)$ for the solution of \eqref{qs1}, and $u_{\gamma_0,\gamma_1}(x,t)$ for the solution of \eqref{Ktwo} with $\lambda_i=m_i/\gamma_i$, as $i=0,1$, i.e.
\begin{align*}\tag{\ref{Ktwo}}
\left\{\begin{array}{rlr}
 \partial_tu_{\gamma_0,\gamma_1} = &\!\!\!\! - \s x (1\!-\!x) \partial_xu_{\gamma_0,\gamma_1} - m_1 f_1 x \, \Jcal\!\left(u_{\gamma_0,\gamma_1},-\gamma_1 x\right) & 
\\[.15cm]
&\!\!  +m_0 f_0 (1\!-\!x) \, \Jcal\!\left(u_{\gamma_0,\gamma_1}, \gamma_0(1\!-\!x)\right)  , &  x\in[0,1], \, t>0 ,\\[.15cm]
u_{\gamma_0,\gamma_1}(x,0)=&\!\!\!\! x, & x\in[0,1],\, t=0. \end{array}\right.
\end{align*}

\subsection{Constant fitness: a modified quasispecies equation}\label{sec:Two_qs}
A detailed analysis has been carried out in \cite{ACNT} for the point-mutation model in the quasispecies case, where the fitness functions (and then also the selection spread) are constant.
To fix the ideas we take 
\[ \s>0, \] 
namely we label 1 the specie with lower fitness. In this setting it is not hard to show that the functions $u_{\gamma_0,\gamma_1}$ are convex w.r.t.~$x$ \cite[Lemma 4.2]{ACNT}. This fact has two relevant consequences.
Firstly,  point  mutations increase the survival opportunities of the low-fitness species, for any choice of $\gamma_0, \gamma_1$. Next, the family of functions $u_{\gamma_0,\gamma_1}$ is ordered both w.r.t.~$\gamma_0$ and $\gamma_1$.
\begin{proposition}\label{teogamma}
For any $(\gamma_0,\gamma_1)$, we have that $u_{\gamma_0,\gamma_1}(x,t)\ge v(x,t)$ for all $x,t$.
\\
For every $(x,t)$, the function $(\gamma_0,\gamma_1)\mapsto u_{\gamma_0,\gamma_1}(x,t)$ is continuous and nondecreasing w.r.t.~to both $\gamma_0$ and $\gamma_1$. 
\end{proposition}
The first statement is given in  \cite[Propositions  4.2]{ACNT}. The second one  is a straightforward extension of \cite[Proposition 4.6]{ACNT}. A relevant issue is the asymptotic behavior for  large time. 
It is well known that the replicator-mutator equation \eqref{qs1} has a constant asymptotic equilibrium  at $\bar x\in[0,1]$ singled out by the relation 
\begin{equation*}\label{qsequi}
\s \bar x(1-\bar x) - m_0 f_0(1-\bar x) + m_1 f_1 \bar x=0 .
\end{equation*}
We thus ask whether also the family of point mutation equations have a constant equilibrium and, in positive case, if this equilibrium is the same of the standard replicator-mutator, or rather depends by the value of the parameters $\gamma_i$. 
To give precise statements, let us split the range of the parameters $m_0, m_1$ into four sub-ranges:
\begin{itemize}
\item[$E:$]
$=\{0\}\times(0,1]$, the extinction range. Here  mutation is fair, because the mutated descendants have higher fitness than their progenitors. 
So mutation helps selection in fixing the higher type, and  $\bar x=0$ is a globally stable equilibrium for  the replicator-mutator equation.
\item[$F:$]
$=[\s/f_0,1]\times \{0\}$, the fixation range. Here mutation is unfair and happens with high probability, so that it is able to overwhelm selection and the high-fitness species extinguishes. According to the replicator-mutator model, $\bar x= 1$ is a globally stable equilibrium.
\item[$C:$]
$=(0,\s/f_0)\times\{0\} \cup (0,1]^2$, the coexistence range, where  both species survive and the equilibrium $\bar x\in(0,1)$.
We further distinguish two different sub-ranges.
\begin{itemize}
\item[$C_0:$]
$=(0,\s/f_0)\times\{0\}$. Here mutation is always unfair, but the quantity of mutated offspring is not sufficient to get rid of selection. 
An unstable equilibrium shows up at $x=1$, so that the basin of attraction of $\bar x=m_0f_0/\s$ is the set $[0,1)$ and $v(1,t)\equiv 1$.
\item[$C_1:$]
$=(0,1]^2$. Here mutation can go backwards and forwards, and $\bar x\in(0,1)$ is a globally stable equilibrium.
\end{itemize}
\end{itemize}

Let us explicitly remark that the replicator-mutator equilibrium $\bar x$  is globally stable for all values of the parameters $m_i$ outside $C_0$.
From a qualitative point of view, the whole family of equations \eqref{Ktwo} behave similarly. We summarize in the following proposition some results obtained in \cite{ACNT}.

\begin{proposition}[Asymptotic stability]\label{largetwo}
\begin{enumerate}[(i) ]
\item  If $(m_0,m_1) \in E\cup F\cup C_1$, then for any value of the parameters $\gamma_0,\gamma_1$  there exists $\bar u (\gamma_0,\gamma_1)\in[\bar x, 1]$ a globally stable equilibrium for equation \eqref{Ktwo}. Actually $u_{\gamma_0,\gamma_1}(x,t)$ converges to $\bar u(\gamma_0,\gamma_1)$ as $t\to\infty$, uniformly w.r.t.~$x\in[0,1]$. Moreover the speed of convergence to equilibrium is exponential, uniformly w.r.t.~$\gamma_o,\gamma_1$.
\item Suppose that $(m_0,m_1)\in C_0$ and $\gamma_0\in(0,1)$ is different from  the unique solution to the equation
\begin{equation}\label{eq_gammaStar}
\s\gamma+m_0f_0\log{(1-\gamma)}=0.
\end{equation}
Then there exists $\bar u(\gamma_0,0)\in[\bar x, 1)$ a locally stable equilibrium for equation \eqref{Ktwo}. Precisely $u(1,t)=1$ for every $t$, while $u_{\gamma_0,0}(x,t)$ converges to $\bar u(\gamma_0,0)$ as $t\to\infty$, uniformly w.r.t.~$x$ in any set of $[0,1-\delta]$ with $\delta\in(0,1)$. \\
Moreover the convergence rate to equilibrium is controlled by the estimate
\begin{equation}\label{unfair:ux_teo_estimate}
\left|\partial_tu(x,t)\right|\leq \frac{\const(\delta)}{(1-\gamma_0)^{\alpha}} e^{-\beta t} ,
\end{equation}
where $\alpha$ and $\beta$ are respectively the maximum point and the maximum value of the concave function
\begin{equation}\label{alphabeta}
\beta(\alpha)=  m_0f_0 \dfrac{1-(1-\gamma_0)^{1-\alpha}}{\gamma_0}+\s (\alpha-1), \qquad 0\le\alpha\le 2. 
\end{equation}
\end{enumerate}
\end{proposition}

\begin{remark}\label{remark:gamma*}
 If $(m_0,m_1)\in C_0$ and $\gamma^*$ satisfies \eqref{eq_gammaStar}, then we are only able to prove that $u_{\gamma_0,0}(x,t)$ converges to a continuous, nondecreasing and convex function $\bar u (x)$, as $t\to\infty$, uniformly w.r.t.~$x$ in any closed subset of $[0,1)$. 
We conjecture that, even though the arguments of the proof of item {\emph{(ii)}} do not apply, the same conclusion holds also in this case. 
In the following we report numerical simulations that go into this direction, see Section \ref{3}.
\end{remark}

\begin{remark}\label{conv_rate}
The estimate of $\partial_tu$ \eqref{unfair:ux_teo_estimate} is obtained by plugging into equation \eqref{Ktwo} the estimate of $\partial_xu$ obtained in \cite[proof of Proposition 4.7]{ACNT}. The constant $\const(\delta)$ can be accurately computed. 
In particular, when interested in small values of $x$ (precisely for $(\s-m_0f_0)/\s<\delta <1)$), we have $\const(\delta)= 2m_0f_0 \, \max\{1,\delta^{1-\alpha}\}$. We shall use this constant in the following chapter, to  compare the theoretical asymptotic rate \eqref{unfair:ux_teo_estimate} with the numerical one.
\end{remark}

We next address to a more quantitative aspect, and ask  if the equilibria of the point type mutations, $\bar u(\gamma_0,\gamma_1)$, actually depend on $\gamma_0,\gamma_1$ and differ from the equilibrium of standard mutations, $\bar x$.
We examine separately the four sub-regions.

In the extinction region $E$, the function $u_{0,\gamma_1}$ stays between the pure selection replicator model \eqref{r1}, and the selection-mutation model \eqref{qs1} (see \cite[Proposition 4.4]{ACNT}). Therefore $\bar u(0,\gamma_1)=0=\bar x$.

In the fixation region $F$, the high-fitness species extinguishes according to both replicator-mutator and point-mutation models, i.e.~$\bar u(\gamma_0,0)=1=\bar x $ for any $\gamma_0$ (see \cite[Proposition 4.5]{ACNT}). 

A new scenario arises in $C$, where the standard  mutator-replicator model provides that the  two species coexist. The same holds for the point-mutation model, although the composition of the mixed population at equilibrium is different. The concentration in time of mutations, exhibited by the point-process model, favours the low fitness species. When $(\gamma_0,\gamma_1)$ approaches $(0,0)$ (continuously distributed mutations), the frequency of the lower fitness type at equilibrium gets its minimum, which is the exact value $\bar x$ of the standard quasispecies equation. On the other side, when $(\gamma_0,\gamma_1)$ approaches $(1,1)$ (very concentrated mutations), it reach its maximum, given by ${m_0f_0}/(m_0f_0+m_1f_1) > \bar x$.
It is remarkable that the asymptotic equilibrium is $\bar u_{1,0}=1$ for $(m_0,m_1)\in C_0$ and $\gamma_0=1$: in this limit case, the time concentration of mutations is sufficient to overwhelm selection and cause the extinction of the type with higher fitness.
Precisely we have that
\begin{proposition}\label{gammateo}
Take $(m_0,m_1)\in C_0$, i.e.~$0<m_0<\s/f_0$ and $m_1=0$.
For every $(x,t)$, the function $(0,1]\ni\gamma_0\mapsto u_{\gamma_0,0}(x,t)$ is nondecreasing and continuous, with
\[ \lim\limits_{\gamma_0\to 0} u_{\gamma_0,0}(x,t)=v(x,t)  \qquad \text{and} \qquad  \lim\limits_{\gamma_0\to 1} u_{\gamma_0,0}(x,t)=u_{1,0}(x,t). \]
Concerning the asymptotic equilibrium, we have
\[ \lim\limits_{\gamma_0\to 0} \bar u(\gamma_0,0)= \bar x =\frac{m_0f_0}{\s}  \qquad \text{and} \qquad \lim\limits_{\gamma_0\to 1} \bar u(\gamma_0,0)=1. \]
Moreover 
\begin{equation}\label{teo_bound}
\bar u(\gamma_0,0) \le \bar{x}/(1-\gamma_0)
\end{equation}
for all $\gamma_0\in(0,1-\bar x)$.
 \end{proposition}
The  first two sentences have been proved, respectively, in \cite[Propositions 4.6 and 4.8]{ACNT}. The estimate \eqref{teo_bound} follows by  a slight variation of the proof of \cite[Proposition 4.8]{ACNT}.
\remove{For $\gamma_0\in(0,1-\bar x)$, we take into account  the solution to \eqref{qs1} with parameters $m_0=\lambda_{0} \gamma_0 /(1-\gamma_0)$, $m_1=0$, namely the function
\[ \hat v(x,t)= \hat x  + \dfrac{x-\hat x}{1+\dfrac{1-x}{1-\hat x}\left(e^{\s(1-\hat x)t}-1\right)} ,\]
where $\hat x = \bar{x}/(1-\gamma_0)$.
Computing equation \eqref{Ktwo} along $\hat v$ gives
\begin{align*}
\partial_t\hat v + \s x (1-x) \partial_x\hat v -  \lambda_{0} f_0 \left[ \hat v(x+\gamma(1-x),t) - \hat v(x,t) \right] \\
= \s \hat x \left[ \gamma(1-x)  \partial_x\hat v - (1-\gamma)\left(\hat v(x+\gamma(1-x),t) - \hat v(x,t) \right)\right] \\
=\dfrac{ \s \gamma \hat x  (1-x)e^{\s(1-\hat x)t}}{\left(1+\dfrac{1-x}{1-\hat x}\left(e^{\s(1-\hat x)t}-1\right)\right)^2\left(1+(1-\gamma)\dfrac{1-x}{1-\hat x}\left(e^{\s(1-\hat x)t}-1\right)\right)} \ge 0,
\end{align*} 
and therefore by comparison $u_{\gamma_0,1}(x,t)\le \hat v (x,t)$ for every $x$ and $t$. 
The thesis follows because for any $x<1$ we have $\bar u(\gamma_0,1)= \lim\limits_{t\to\infty}u_{\gamma_0,0}(x,t) \le \lim\limits_{t\to\infty}\hat v (x,t)=\hat x$.
}
These results can be easily extended to  the internal region $C_1$, where mutations can happen from 0 to 1 and back from 1 to 0.  This latter case, apparently more general than the previous one, is actually easier to handle because we have a global equilibrium and $x$-continuity is preserved when $t\to+\infty$.
\begin{proposition}Take $(m_0,m_1)\in C_1$.
For every $(x,t)$, the function $(0,1]^2\ni (\gamma_0,\gamma_1)\mapsto u_{\gamma_0,\gamma_1}(x,t)$ is nondecreasing and continuous, with respect to $\gamma_0$ and $\gamma_1$, separately. Moreover
\[\begin{array}{cc}
 \lim\limits_{(\gamma_0,\gamma_1)\to (0,0)} u_{\gamma_0,\gamma_1}(x,t)=v(x,t)  , \qquad & \qquad \lim\limits_{(\gamma_0,\gamma_1)\to (1,1)} u_{\gamma_0,\gamma_1}(x,t)=u_{1,1}(x,t).
\end{array}\]
Concerning the asymptotic equilibrium, we have
\[\begin{array}{cc}
 \lim\limits_{(\gamma_0,\gamma_1)\to (0,0)} \bar u_{\gamma_0,\gamma_1}=\bar x  , \qquad & \qquad \lim\limits_{(\gamma_0,\gamma_1)\to (1,1)} \bar u_{\gamma_0,\gamma_1}=\dfrac{m_0f_0}{m_0f_0+m_1f_1}.
\end{array}\]
\end{proposition}

These analytical arguments do not allow to compute the asymptotic equilibrium $\bar u(\gamma_0,\gamma_1)$, nor to see if it depends continuously by the parameters $\gamma_0, \gamma_1$. It could also happen, at the contrary, that there is a bifurcation value which separates a set of models that converge to the quasispecies equilibrium $\bar x$ from another one which brings to extinction of the high-fitness specie. 
The numerical simulations, produced in the following chapter, suggest that the asymptotic equilibrium  $\bar u(\gamma_0,\gamma_1)$ spans  the segment line between $\bar x$ and ${m_0f_0}/({m_0f_0+m_1f_1})$.

\subsection{Density dependent fitness.}\label{sec:density_dependent_fitness}

In evolutionary theory, the fitness functions $f_0$ and $f_1$ are assumed to depend on the population density  in a linear way.
A ``payoff matrix'' with nonnegative entries is introduced
\[ A = \left(\begin{array}{cc} a_{0} & b_{0} \\ a_{1} & b_{1}\end{array}\right),
\]
and the fitness functions are defined by means of 
\[ f_i(x)= a_{i}(1-x)+ b_{i}x, \quad \mbox{ as $i=0,1$.}\] 
This makes the analytical study of equation  \eqref{Ktwo} much more complicate: well-posedness established in \cite{ACNT} still applies, but the qualitative results concerning quasispecies do not extend in general.
In particular, the selection spread $\s(x)= f_0(x)-f_1(x)$ depends itself by $x$, therefore it possibly changes sign and there is not a clear separation between fair and unfair mutations.
The following comparison result states that point-type mutations can even punish the low-fitness type, if compared with the standard replicator-mutator. This can happen when the global amount of mutations is not enough to balance the increasing of the fitness of the other type.

\begin{proposition}\label{uvcomp}
If 
\[ m_0 (b_0-a_0) + m_1 (b_1-a_1) \ge \max \{ 2(a_0-a_1)-(b_0-b_1), -(a_0-a_1)+2(b_0-b_1)\},\]
 then $u_{\gamma_0,\gamma_1}\ge v$  pointwise, for any value of $\gamma_0$ and $\gamma_1$.
On the contrary, if 
\[ m_0 (b_0-a_0) + m_1 (b_1-a_1) \le \min  \{ 2(a_0-a_1)-(b_0-b_1), -(a_0-a_1)+2(b_0-b_1)\},\]
then $u_{\gamma_0,\gamma_1}\le v$  pointwise, for any value of $\gamma_0$ and $\gamma_1$.
\end{proposition}
\begin{proof}
We compute the equation in \eqref{Ktwo} along $v$, the solution to \eqref{qs1}, and get
\begin{align}\nonumber 
\partial_tv + \s x (1-x) \partial_xv  -f_0\jcal_0 v- f_1\jcal_1 v = \\ \label{XinK}
-\dfrac{m_0f_0}{\gamma_0} T_1 v (x,\gamma_0(1-x),t) -\dfrac{m_1f_1}{\gamma_1} T_1 v (x,-\gamma_1x,t),
\end{align}
where $T_1$ stands for the first order Taylor expansion w.r.t.~$x$:
\[ T_1 v(x,z,t)= v(x+z,t)-v(x,t)-z \partial_x v(x,t) = \dfrac{1}{2} z^2 \int_0^1 \partial^2_{xx} v(x+\tau z,t)\, d\tau .\]
Besides, deriving equation \eqref{qs1} w.r.t.~$x$ and applying comparison principle assures that $\partial_x v \ge 0$. Next, deriving again  gives that $w=\partial^2_{xx} v$ solves
\[\left\{\begin{array}{ll}
\partial_t w + a(x) \partial_x w +2 a'(x) w= a'' \partial_x v , \qquad & 0\le x\le 1 , \, t>0 \\
w(x,0)=0 , & 0\le x\le 1 , \, t=0, 
 \end{array}\right.\]
with $a(x)=\s x(1-x) -m_0f_0(1-x)+m_1f_1 x$.
Hence if $a$ is convex, then $0$ is a subsolution  and therefore $\partial^2_{xx}v\ge 0$.
Coming back to \eqref{XinK}, we see that $v$ is a  subsolution to \eqref{Ktwo} and therefore $v\le u_{\gamma_0,\gamma_1}$. The opposite happens if $a$ is concave. 
On the other hand, by construction $a''$ is a linear function of $x$, so the thesis readily follows by computations.
\end{proof}

The particular case of quasispecies can be recovered by taking constant fitness, i.e. $b_0-a_0=b_1-a_1=0$. Hence the first part of Proposition \ref{uvcomp} always holds, provided that  the type with lower fitness has been labeled $0$. 

\begin{remark}\label{constss}
Another interesting particular case stands in taking a constant selection spread, i.e. $ a_0-a_1=b_0-b_1=\s$. Here Proposition \ref{uvcomp} gives a complete picture and states that it $ m_0 (b_0-a_0) + m_1 (b_1-a_1) > \s$, then $u_{\gamma_0,\gamma_1}> v$, while if $ m_0 (b_0-a_0) + m_1 (b_1-a_1) < \s$, then $u_{\gamma_0,\gamma_1}> v$. When  $m_0 (b_0-a_0) + m_1 (b_1-a_1) = \s$, then $u_{\gamma_0,\gamma_1}= v$  for any values of the $\gamma_i$'s.
\end{remark}

Up to now, Proposition \ref{uvcomp} is  the only theoretical tool we have in hands to analyze density dependent models. This is the reason we decided perform a series of numerical investigations. In subsection \ref{sec:pd} we shall focus on the Prisoner's Dilemma.


\section{Numerical assessment}\label{3}
In this Section we shall discuss the numerical approximation and some numerical tests for the models 
that have been described above. 
We shall start from the case of the simple replicator equation \eqref{r1}. From this first example it is indeed possible to detect the major difficulties on the approximation of the models we are looking at.

Let start by defining an uniform grid on the set $[0,1]\times[0,T]$. For $\Dx,\Dt\in\R^+$, we define the grid as the set of points $\gcal=\{(x_j,t^n)=(j\Dx,n\Dt), \ j=0,\ldots,N_x, \  n=0,\ldots,N_t\}$ with $N_x\Dx=1$ and $N_t\Dt=T$. A standard approximation of \eqref{r1} on $\gcal$, reads as 
\begin{equation}\label{noSource:scheme_std}
u(x_j,t^{n+1})= \left(1-\Frac{a(x_j)\Dt}{\Dx}\right) u(x_j,t^{n}) + \Frac{a(x_j)\Dt}{\Dx} u(x_{j-1},t^{n}),
\end{equation}
with initial and boundary conditions $u(x_j,0)=x_j$ and $u(0,t^n)= 0$ respectively.

By standard computation, it is possible to point out that the numerical error estimate substantially depends on the behaviour of the space derivatives of the exact solution $u$. 
For the specific problem \eqref{r1}, the exact solution 
\begin{equation}\label{noSource_sol_esatta}
u(x,t) = \Frac{x e^{-\s(t-t_0)}}{1-x(1-e^{-\s(t-t_0)})}
\end{equation}
can be computed by the method of characteristics.                    
It is then easy to check that, for every $0\leq x<1$ be fixed,
$$ \lim_{t\rightarrow+\infty} u(x,t) = 0 \mbox{ and } \lim_{t\rightarrow+\infty} \partial_x u(x,t) = 0,$$
while for all $t>0\ $ $u(1,t)=1$. Moreover, as time grows, $\partial_x u(1,t)$ grows exponentially in time giving a loss of accuracy on the approximation \eqref{noSource:scheme_std}.

We note that, for some set of problem parameters, we will get the same kind of behaviour solving the more general problem \eqref{Ktwo}. Moreover a second difficulty arises due to the non-local nature itself. Looking for instance at the quantity $\Jcal(u,\gamma_0(1-x))$ defined in \eqref{j0}, 
the point is to compute the non-local value $u(x_j+\gamma_0(1-x_j),t^n)$ by means of the known nodes values $u(x_k,t^n)$, $k=0,\ldots,N_x$. The simplest idea is of course to use a linear interpolation. 
However, by the same arguments as before, we must be careful when the non-local point falls in the last cell close to $x=1$.
Indeed, as it has been shown in the previous section \ref{sec:Two_qs}, for this problem too, the first derivative 
$\partial_x u(1,t)$ blows up and 
the approximation by means of linear interpolation loses of accuracy close to the boundary node $x=1$, giving rise to incorrect solutions also in the interior of the domain. It then turns out to be more accurate to compute the value $u(x_j+\gamma_0(1-x_j),t^n)$ by an extrapolation between the two last internal grid points (see Section \ref{sec:unfair}, Figure \ref{fig:unfair_interp_extrap}). At the same time, the extrapolation needs a fine grid to ensure accuracy.
We will handle these difficulties by adopting a variable numerical grid which thickens with the increasing of time near the right boundary $x=1$, and giving an ad-hoc approximation of the nonlocal term $\Jcal(u,\gamma_0(1-x))$.

\subsection{Replicator equation}\label{sec:noSource}
Here we look at the numerical solution of the replicator equation \eqref{r1}. 
As we have already mentioned, to compute the numerical solution we shall adopt a variable grid which thickens with the increasing of time near the right boundary $x=1$. 
\\
Let us first fix the total number of nodes $N_x$ and assume to have on the plane $(x,t)$ at each time $t^n$, $n=0,\ldots,N_t$, the generic set of points $\gcal^n = \{x^n_j,\ j=0,\ldots,N_x\}$, where the space step $\Dx^n_j=x^n_{j}-x^n_{j-1}$ is not constant, $x^n_0=0$, $x^n_{N_x}=1$ $\forall n$, and the time step $\Dt^n=t^{n+1}-t^n$ varies according to a monotonicity condition.
We then define our numerical approximation as 
\begin{equation}\label{noSource:scheme_alpha}
u(x^{n+1}_j,t^{n+1})= (1-\alpha^n_j) u(x^{n}_j,t^{n}) + \alpha^n_j u(x^{n}_{j-1},t^{n}),
\end{equation}
where
\begin{equation}\label{alpha}
\alpha^n_j = \Frac{a(x^n_j)\Dt^n- (x^{n+1}_j-x^n_j)}{\Dx^n_j}.
\end{equation}
Notice that for $x^{n+1}_j = x^n_j$, scheme \eqref{noSource:scheme_alpha} reduces to the standard upwind scheme \eqref{noSource:scheme_std}.

We shall now describe how to construct the variable space grid of  at each time step. 
We start by fixing the total number of nodes $N_x$ in $[0,1]$. The main idea is to increase the number of nodes in the region where the space gradient is higher. To this aim, we shall use the fact that our solutions are monotone increasing and convex.  Since $\partial_x u(x,0)=1$ $\forall\ x\in[0,1]$, at each time $t^n$ we select the point $X(t^n)$ such that
\begin{equation}\label{def_X}
X(t^n)\in(0,1]: \quad \forall\  x^n_j<X(t^n) \quad \mbox{ it holds } \quad  u(x^n_j,t^n) -u(x^n_{j-1},t^n)<\Dx^n_j .
\end{equation}
Then, at each time step $t^n$, the interval $[0,1]$ is divided into two parts, $[0,X(t^n))$ and $[X(t^n),1]$. Now, fix $N_l$ and $N_r$, with $N_r>N_l$ and $N_r+N_l=N_x$. As time grows, the solution on the left side, defined by the points for which the numerical space derivative is less than one, 
does not require a fine approximation and will be computed on a grid with a few number of nodes $N_l$. On the other hand, the solution on the right part needs a more accurate approximation and will be computed with a higher number of nodes $N_r$ (see for instance Figure \ref{fig:noSource_griglia}). Then we set
\begin{equation}\label{passi_small_big}
\Dx^n_{big} = X(t^n)/N_l, \quad\mbox{ and }\quad \Dx^n_{small}=(1-X(t^n))/N_r.
\end{equation}

This procedure together with the expression \eqref{alpha}, defines our numerical scheme. Henceforth, we shall refer to it as \textit{Adaptive grid scheme}
opposed to the \textit{Uniform grid scheme} defined by \eqref{noSource:scheme_std}.

To satisfy the monotonicity request it is needed that $0\le \alpha^n_j \le 1,$
which yields
\[ (x^{n+1}_j-x^n_j) \leq a(x^n_j)\Dt^n \leq \Dx^n_j + (x^{n+1}_j-x^n_j)\]
for all $x^n_j, x^{n+1}_j\in\ (0,1)$.

In the numerical tests, we shall compare the results obtained by applying the 
\textit{Uniform grid scheme} and the \textit{Adaptive grid scheme}.
For the first one we shall fix the space step $\Dx$ and consequently the time step
\begin{equation}
 \Dt =\min_{x_j\in(0,1)} \Frac{\Dx}{a(x_j)},
\end{equation}
while for the second one, the space step varies and hence the time step varies too according to the rule
\begin{equation}
 \Dt^n =\min_{x^n_j,x^{n+1}_j\in(0,1)} \Frac{\Dx^n_j+(x^{n+1}_j-x^n_j)}{a(x^n_j)}.
\end{equation}

In Figure \ref{fig:noSource_confr_DxBig_T5}, the numerical error is given as a function of time and space, and it has been computed by the formula
\begin{equation}\label{noSource_error}
e(x_j^n,t_n) = |u^n_j-u(x_j,t^n)|, 
\end{equation} 
where $u(x_j,t^n)$ is the exact solution given by \eqref{noSource_sol_esatta} and $u^n_j$ is the numerical solution given by the two schemes that we are looking at.

We also compare the two numerical schemes with different number of nodes $N_x$. Specifically, the \textit{Uniform grid scheme} is computed fixing the space step as $\Dx=\Dx_{big}=\max_{n=0,\ldots,N_t}\Dx^n_{big}$, where $\Dx^n_{big}$ have been defined in \eqref{passi_small_big}. This test shows that
\textit{Adaptive grid scheme} gives globally a better approximation.  In particular, from the error graph in Figure \ref{fig:noSource_confr_DxBig_T5}-(b) we observe that the quality of the 
numerical solution given by the \textit{Adaptive grid scheme} is better for $x<1$ and then it is highly accurate in capturing the asymptotic constant value $u(x,+\infty)\equiv 0$ for $x<1$. Such behavior may also be seen zooming the graph near at $x=0$, as done in \ref{fig:noSource_confr_DxBig_T5}-(c).
\\
These same conclusions are confirmed in Figure \ref{fig:noSource_confr_N31_T5}, where the two numerical approaches with the same total number of space nodes $N_x$ are compared at final $T=5$. The difference between the two schemes is then mainly on the displacement of nodes on the interval $[0,1]$. 
It is clear from graphs that in both cases \textit{Adaptive grid scheme} keeps on giving a better approximation. As in the previous test, from the error graphs in Figure \ref{fig:noSource_confr_N31_T5}-(b) 
we observe that, for $x<1$, the quality of the 
numerical solution given by the \textit{Adaptive grid scheme} is still better than the one given by the \textit{Uniform grid scheme} and again it is highly accurate in capturing the asymptotic constant value.

\begin{figure}[!htb]
\begin{center}
\begin{tabular}{c}
\includegraphics[scale=0.8]{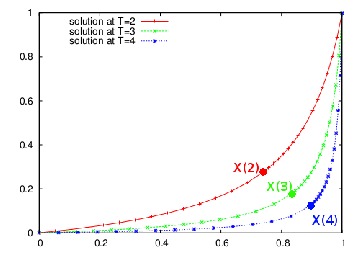}
\end{tabular}
\end{center}
\caption{Test for Section \ref{sec:noSource}. Example of \textit{Adaptive grid}. The $(+)$-line is the solution computed at time $T=2$, the 
(x)-line is the solution computed at time $T=3$ and the $(*)$-line is the solution computed at time $T=4$. The values $X(2)$, $X(3)$, $X(4)$ are the points 
defined in \eqref{def_X} that split the interval $[0,1]$ into two parts, the left one where $\partial_x u(\cdot,\cdot)\leq 1$ and the right one where $\partial_x u(\cdot,\cdot)> 1$ respectively. Notice that $X(4)>X(3)>X(2)$ and that, for construction,  
in all cases grid nodes are more dense after $X(\cdot)$. Here, $X(T=2)\approx 0.741$, $X(T=3)\approx0.835$, $X(T=4)\approx 0.895$.} 
\label{fig:noSource_griglia}
\end{figure}

\begin{figure}[htb!]
\begin{center}
\begin{tabular}{cc}
\includegraphics[scale=0.5]{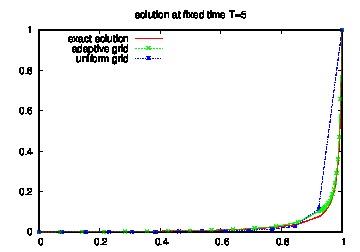}
&
\includegraphics[scale=0.5]{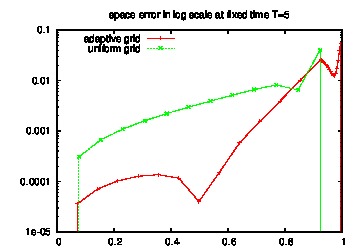}
\\
(a) & (b)
\\
\includegraphics[scale=0.5]{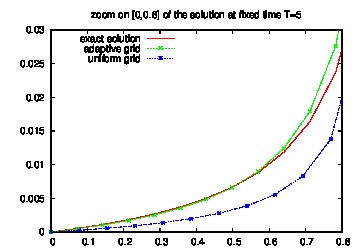}
&
\includegraphics[scale=0.5]{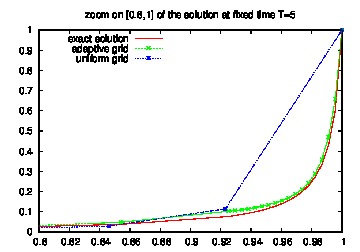}
\\
(c) & (d)
\end{tabular}
\end{center}
\caption{Test for Section \ref{sec:noSource}. Comparison among the two solutions obtained by applying \textit{Adaptive grid scheme} \eqref{noSource:scheme_alpha}-\eqref{alpha} and \textit{Uniform grid scheme} \eqref{noSource:scheme_std} with $\Dx_{big}=0.07$ at final time $T=5$. Figure (a): The (x)-line is the solution of \textit{Adaptive grid scheme}, the ($*$)-line is the solution of \textit{Uniform grid scheme}. The two solutions are compared to exact solution \eqref{noSource_sol_esatta} in smooth line ($-$). Figure (b): graph of the two errors \eqref{noSource_error} in log-scale. Figure (c): zoom in on the solutions in the subinterval  $[0,0.8]$, Figure (d): zoom in on the solutions in the subinterval  $[0.8,1]$. 
From the error graph (b) we observe that for low values of $x$, the quality of the 
numerical solution given by the \textit{Adaptive grid scheme} is better than the one given by the \textit{Uniform grid scheme} and then it is highly accurate in capturing the asymptotic constant value $u(x,+\infty)\equiv 0$ for $x<1$. Such behaviour may also be seen in (c), where we focus on subinterval $[0,0.8]$.
} 
\label{fig:noSource_confr_DxBig_T5}
\end{figure}

\begin{figure}[htb!]
\begin{center}
\begin{tabular}{cc}
\includegraphics[scale=0.5]{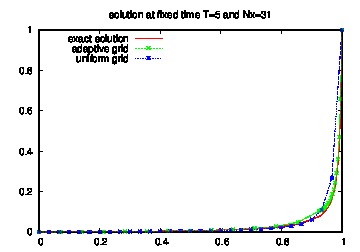}
&
\includegraphics[scale=0.5]{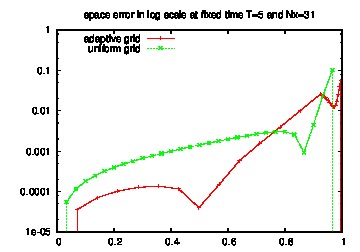}
\\
(a) & (b)
\\
\includegraphics[scale=0.5]{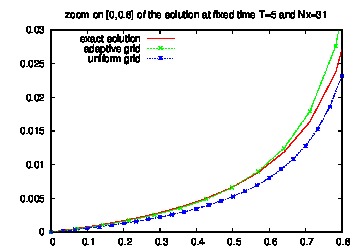}
&
\includegraphics[scale=0.5]{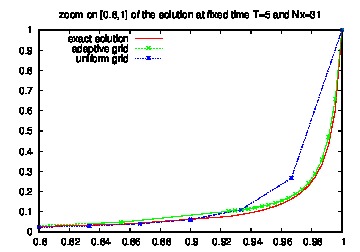}
\\
(c) & (d)
\end{tabular}
\end{center}
\caption{Test for Section \ref{sec:noSource}. Comparison among the solutions obtained by applying \textit{Adaptive grid scheme} \eqref{noSource:scheme_alpha}-\eqref{alpha} and \textit{Uniform grid scheme} \eqref{noSource:scheme_std} with the same total number of nodes $N_x=31$ at final time $T=5$. Figure (a): The (x)-line is the solution of \textit{Adaptive grid scheme}, the ($*$)-line is the solution of \textit{Uniform grid scheme}. The two solutions are compared to exact solution \eqref{noSource_sol_esatta} in smooth line ($-$). Figure (b): graph of the two errors \eqref{noSource_error} in log-scale. Figure (c): zoom in on the solutions in the subinterval $[0,0.8]$. Figure (d): zoom in on the solutions in the subinterval $[0.8,1]$.} 
\label{fig:noSource_confr_N31_T5}
\end{figure}

\begin{figure}[!htp]
\begin{center}
\begin{tabular}{c}
\includegraphics[scale=0.5]{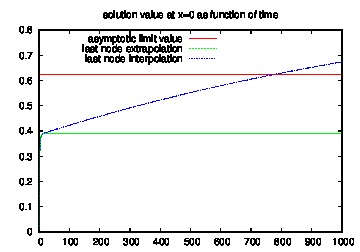}
\end{tabular}
\end{center}
\caption{Unfair mutation with $\gamma_0=0.6$ and $\bar x=0.25$: solution value at point $x=0$ as function of time. The red line is the theoretical asymptotic limit $\bar x/(1-\gamma_0)$ obtained in Proposition \ref{gammateo}. The green and the blue lines are the asymptotic value obtained by scheme \eqref{source:scheme}; for the green one when the non-local value $x^n_j+\gamma_0(1.-x^n_j)$ falls close to $x=1$ it has been computed by an extrapolation with respect the two last internal grid nodes, while for the blue one it has been computed by an interpolation with respect the last grid node and the boundary node $x=1$. As it has been described in Section \ref{sec:unfair}, the use of the interpolation leads to an incorrect solution that grows asymptotically in time.} 
\label{fig:unfair_interp_extrap}
\end{figure}

\subsection{Point-type quasispecies  with weak unfair mutation ($(m_0,m_1)\in C_0$))}\label{sec:unfair}
Here we shall focus on problem \eqref{Ktwo} with $(m_0,m_1)\in C_0=(0,\s/f_0)\times\{0\}$.

We shall apply the \textit{Adaptive grid numerical} approach introduced in Section \ref{sec:noSource}.  The differential part of \eqref{Ktwo} is then approximated by scheme
\eqref{noSource:scheme_alpha}-\eqref{alpha} on the non-Uniform grid $\{x_j^n\}_{n,j}$ defined in \eqref{def_X}.
 We then get
\begin{equation}\label{source:scheme}
\begin{array}{lcl}
u(x^{n+1}_j,t^{n+1})&=& (1-\alpha^n_j) u(x^{n}_j,t^{n}) + \alpha^n_j u(x^{n}_{j-1},t^{n})\\
&&\\
&&-\dfrac{m_0f_0}{\gamma_0}\Dt\left(u(x^{n}_{j},t^{n})-u(x^n_j+\gamma_0(1-x^n_j),t^n)\right),
\end{array}
\end{equation}
under the monotonicity constraint
\begin{equation}
\sup_{x^n_j,x^{n+1}_j\in(0,1)} 
\Frac{x^{n+1}_j-x^n_j}{a(x^n_j)}\leq \Dt^n\leq \min_{x^n_j,x^{n+1}_j\in(0,1)}\Frac{\gamma_0(\Delta x^n_j+x^{n+1}_j-x^n_j)}{\gamma_0a(x^n_j)+m_0f_0\Delta x^n_j} 
\end{equation}
that gives rise to a condition in selecting the Adaptive grid, namely
\begin{equation}
\sup_{x^n_j,x^{n+1}_j\in(0,1)} (x^{n+1}_j-x^n_j) < \Frac{\gamma_0}{m_0f_0}\min_{x^n_j\in(0,1)} a(x^n_j) = \Frac{\gamma_0}{m_0f_0} a(x^n_1),
\end{equation}
at  each time step.

The main point now is to compute the non-local value $u(x^n_j+\gamma_0(1-x^n_j),t^n)$ by means of the known nodes values $u(x_k^n,t^n)$, $k=0,\ldots,N_x$. The simplest idea is of course to use a linear interpolation. 
Specifically, for $k=1,\ldots,N_x-1$ such that
$$
x^n_{k-1} < x^n_j+\gamma_0(1-x^n_j) <x^n_k,
$$
then
$$ 
u^n(x^n_j+\gamma_0(1-x^n_j),t^n)= u^n_k + \Frac{u^n_k-u^n_{k-1}}{x^n_k-x^n_{k-1}}\big((x^n_j+\gamma_0(1-x^n_j))-x_k\big).
$$
However, as we have already mentioned at the beginning of this Section , we must be careful when it happens that the non-local point falls in the last cell close to $x=1$, i.e. 
$$
x^n_{Nx-1} < x^n_j+\gamma_0(1-x^n_j) <x^n_{N_x} = 1.
$$
In this case, it turns out to be more accurate to compute the value $u(x^n_j+\gamma_0(1-x^n_j),t^n)$ by an extrapolation between the two last internal points $x_{N_x-2}$ and $x_{N_x-1}$.
This fact is made clear by Figure \ref{fig:unfair_interp_extrap}, where the solutions obtained by interpolation are compared with the theoretical upper bound \eqref{teo_bound}. 
One can see that the use of the interpolation leads to an incorrect solution that grows asymptotically in time overcoming the theoretical bound value.

\subsubsection{Some qualitative tests for weak unfair mutations} 
Here we consider the weak unfair mutation case, since it is a case where we have some analytical results, but the full behavior of the solutions is not known. We shall give now some numerical tests to highlight the following aspects: 
\begin{enumerate}[(i) ]
\item the better performance of our adaptive numerical scheme with respect to the uniform one; 
\item the dependence of the 
asymptotic equilibrium $\bar u(\gamma_0,0)$ with respect to $\gamma_0$; 
\item the behaviour of the asymptotic equilibrium $\lim_{t\rightarrow \infty} u_{\gamma_0,0}(x,t)=\bar u(x)$ for  $\gamma_0=\gamma^*$, where
$\gamma^*\in(0,1)$ is the only solution to \eqref{eq_gammaStar},
see Remark \ref{remark:gamma*};
\item  the quality of the theoretical estimate given in \eqref{unfair:ux_teo_estimate}.
\end{enumerate}

Let us then go through the points listed above.

\textbf{(i)} The analytical expression of the exact solution to problem \eqref{Ktwo} is not known. 
However, in Proposition \ref{gammateo} some theoretical estimates have been given and they may be used to estimate the goodness of our numerical scheme. 
In Figure \ref{fig:unfair_asym_confr}, we show a test case for which the Adaptive approach gives a numerical solution that meets the theoretical bound \eqref{teo_bound}, while for the same $\Dx_{big}$, the uniform approach gives an asymptotic value that overcome the upper bound value $\bar x/(1-\gamma_0)$.

\textbf{(ii)} It has been shown in Section \ref{sec:Two_qs} that the analytical arguments do not allow to explicitly compute the asymptotic equilibrium $\bar u(\gamma_0,0)$, nor to see if it depends continuously by $\gamma_0$.
Through some numerical simulations we can compute a good approximation of this asymptotic value and we can observe that for every $m_0\in(0,\s/f_0)$, the numerical asymptotic equilibrium $\bar u(\gamma_0,0)$ is a continuous function with respect $\gamma_0$
and we observe that , increasing $\gamma_0$ from $0$ to $1$, the asymptotic equilibrium value $\bar u(\gamma_0,0)$ spans the segment line between $\bar x$ and $1$, see Figure \ref{fig:unfair_sf0_4}-(a)-(b). 

\textbf{(iii)} Proposition \ref{largetwo}-\textit{(ii)} does not apply if $\gamma_0=\gamma^*$. In Figure \ref{fig:unfair_gammaStar} we then plot the solution $u_{\gamma^*,0}(x,t)$ fixing $\s=1$, $m_0=0.4$, $\s/f_0=0.8$ (with this choice of parameters $\gamma^*=0.79$). We then may say that also for the particular value $\gamma^*$ 
the solution $u_{\gamma^*,0}(x,t)$ converges to 
$\bar u(\gamma^*,0)\in [\bar x,1)$ as $t \rightarrow \infty$.

\textbf{(iv)} Here, we would analyze the asymptotic rate by which the solution $u_{\gamma,0}(x,t)$ converges to $\bar u(\gamma_0,0)$. On one hand we compute the solution $u_{\gamma_0,0}$ numerically by scheme \eqref{source:scheme} varying the problem parameters, and we estimate the asymptotic rate as the time $T_{num}(\gamma_0,m_0,\s/f_0)$ such that, for $\epsilon$ small
\begin{equation}\label{num_asym_time}
\partial_t u_{\gamma_0,0}(0,T_{num})\ll \epsilon.
\end{equation}
On the other hand, we compute the two parameter $\alpha$ and $\beta$ in formula \eqref{alphabeta} and we assume the theoretical asymptotic rate as the time $T_{teo}(\gamma_0,m_0,\s/f_0)$ such that
\begin{equation}\label{teo_asym_time}
\frac{\const(\delta)}{(1-\gamma_0)^{\alpha}} e^{-\beta\; T_{teo}}\ll \epsilon \ \mbox{ for } \ \delta = 1 \mbox{ and }  \const(\delta) = 2m_0f_0,
\end{equation}
as it has been expressed in formula \eqref{unfair:ux_teo_estimate} and Remark \ref{conv_rate}.
\\
In Figure \ref{fig:unfair_sf0_8_asymtime} we fix 
$\s/f_0=0.8$  and 
we compare the two surfaces $T_{num}$ and $T_{teo}$ as functions of $\gamma_0$ and $m_0$. We overlap to the surfaces the set of 
values $\{(m_0,\gamma_{max}), \ m_0<\s/f_0\}$ and $\{(m_0,\gamma^*), \ m_0<\s/f_0\}$, where 
$\gamma_{max}$ are the maximum points $\gamma_0$   
of the function $T_{num}(\gamma_0,m_0,\cdot)$.

\begin{figure}[!htp]
\begin{center}
\begin{tabular}{c}
\includegraphics[scale=0.5]{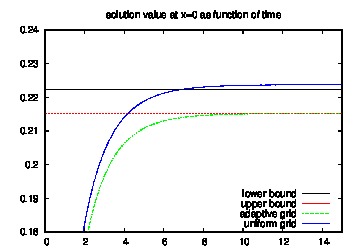}
\end{tabular}
\end{center}
\caption{Unfair mutation with $\gamma_0=0.1$ and $\bar x=0.2$: solution value at point $x=0$ as function of time. The red line is the theoretical asymptotic lower bound $\bar x$ and the black one is the theoretical upper bound $\bar x/(1-\gamma_0)$ given in Proposition \ref{gammateo}. The green and the blue lines are the values obtained by the \textit{Adaptive grid} and the \textit{Uniform grid} respectively. 
The two schemes have been applied with $\Dx_{big}=0.07$. We then get with the adaptive grid a solution that meets the theoretical bounds.}
\label{fig:unfair_asym_confr}
\end{figure}

%

\begin{figure}[htp]
\begin{center}
\begin{tabular}{cc}
\includegraphics[scale=0.4]{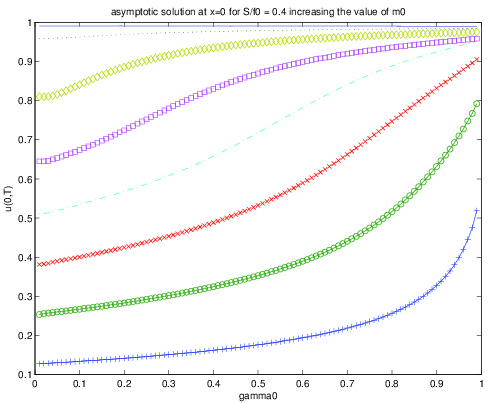}
&
\includegraphics[scale=0.4]{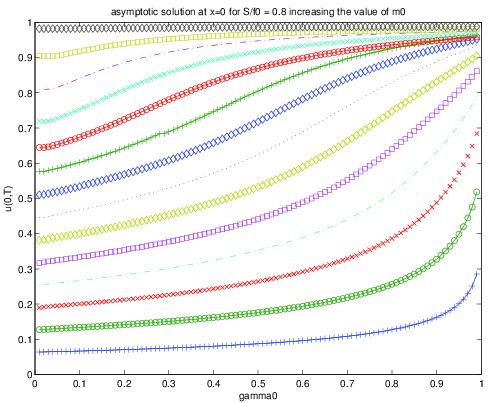}

\\
(a) & (b)
\end{tabular}
\end{center}
\caption{Coexistence region $C_0$ (unfair mutation). Plot of the numerical asymptotic equilibrium $\gamma_0\mapsto\bar u(\gamma_0,0)$, varying $m_0\in(0.05,\s/f_0)$ for  (a) $\s/f_0=0.4$ and (b) $\s/f_0=0.8$ respectively.
We can observe that the asymptotic value is a continuous function with respect to $\gamma_0$, for any value of $m_0$.} 
\label{fig:unfair_sf0_4}
\end{figure}

\begin{figure}[!htp]
\begin{center}
\begin{tabular}{cc}
\includegraphics[scale=0.4]{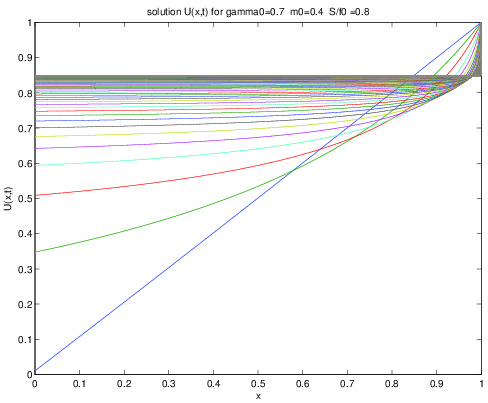}
&
\includegraphics[scale=0.4]{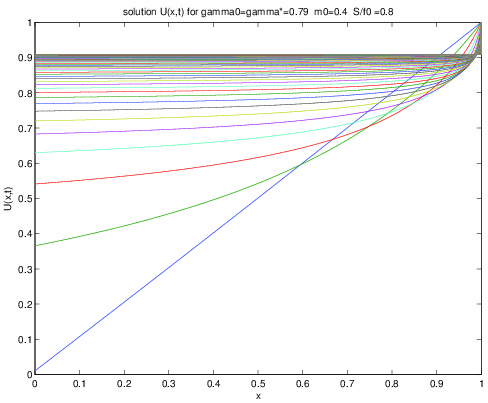}
\\
(a) & (b)
\end{tabular}
\end{center}
\caption{Coexistence region $C_0$ (unfair mutation). Plot of the  numerical solution $u_{\gamma_0,0}(x,t)$, varying the time $t$.
In figure (a)  $\gamma_0=0.7$, while in figure (b) $\gamma_0=0.79$ solves  \eqref{eq_gammaStar}. 
 It is then possible to see that the behaviour of the solution for that particular value  is the same of solution obtained with others $\gamma_0$ values.
} 
\label{fig:unfair_gammaStar}
\end{figure}

\begin{figure}[htp]
\begin{center}
\begin{tabular}{cc}
\includegraphics[scale=0.4]{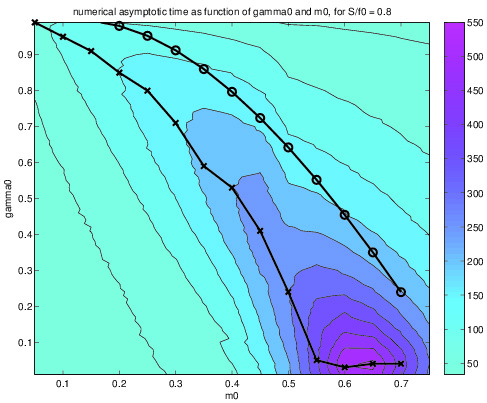}
&
\includegraphics[scale=0.4]{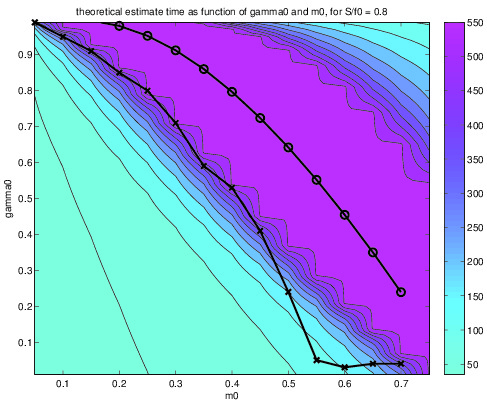}
\\
(a) & (b)
\end{tabular}
\end{center}
\caption{Coexistence region $C_0$ (unfair mutation). 
In both figures, the value of $\s/f_0$ has been fixed at $0.08$. The numerical asymptotic rate $T_{num}$ defined in \eqref{num_asym_time}  (figure (a)), and the theoretical asymptotic rate $T_{teo}$ defined in \eqref{teo_asym_time}  (figure (b)) are plotted on the plane $(m_0,\gamma_0)$. 
In this latter plot we have fixed the surface range of values to $[0,\max_{\gamma_0} T_{num}(\gamma_0,\cdot,\cdot)]$.
The points (x) are the couples  $(m_0,\gamma_{max})$ such that $T_{num}(\gamma_{max},m_0,\s/f_0) = \max_{\gamma_0} T_{num}(\gamma_0,m_0,\s/f_0)$, while the points (o) are the values $(m_0,\gamma^*)$, where $\gamma^*$ is the solution to \eqref{eq_gammaStar}.
}
\label{fig:unfair_sf0_8_asymtime}
\end{figure}

\subsection{Point-type quasispecies  with back and forth mutations ($(m_0,m_1)\in C_1$)}
Here we shall solve numerically problem \eqref{Ktwo} considering various value of $m_0,m_1\in(0,1]$  and  $\gamma_0,\gamma_1\in(0,1]$.
By applying the numerical scheme described above we get
\begin{equation}\label{source:scheme_alpha_beta_hybrid}
\begin{array}{lcl}
u(x^{n+1}_j,t^{n+1})&=& (1-\alpha^n_j) u(x^{n}_j,t^{n}) + \alpha^n_j u(x^{n}_{j-1},t^{n})\\
&&\\
&&+\Dt\dfrac{m_0f_0}{\gamma_0}\left(u(x^n_j+\gamma_0(1-x^n_j),t^n)-u(x^{n}_{j},t^{n})\right)\\
&&\\
&&+ \Dt\dfrac{m_1f_1}{\gamma_1} \left[ u(x^n_j-\gamma_{1}x^n_j,t^n) - u(x^n_j,t^n)\right],
\end{array}
\end{equation}
with $\alpha^n_j$ defined by \eqref{alpha}.

First we fix $\s=1$, $\s/f_0=0.3$,
and  plot the asymptotic value $\bar u_{\gamma_0,\gamma_1}$ in the plane $(\gamma_0,\gamma_1)\in(0,1)^2$, assuming that the couple $(m_0,m_1)$ takes values
into the following set 
$$ S =\{(0.1,0.1),(0.1,0.9),(0.9,0.1),(0.9,0.9)\}.$$
All the graphs in Figure \ref{fig:2qs_test_1} show that the asymptotic value at $x=0$ is a continuous function with respect the couple $(\gamma_0,\gamma_1)$.


\begin{figure}[htp]
\begin{center}
\begin{tabular}{cc}

\includegraphics[scale=0.4]{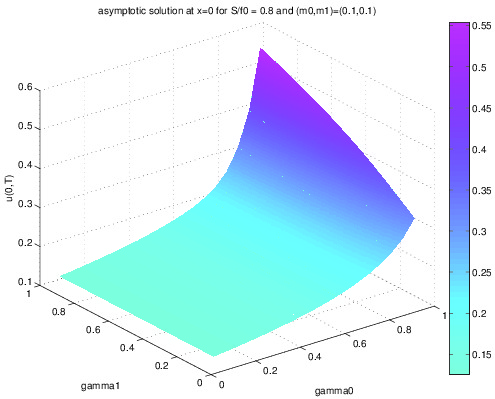}
&
\includegraphics[scale=0.4]{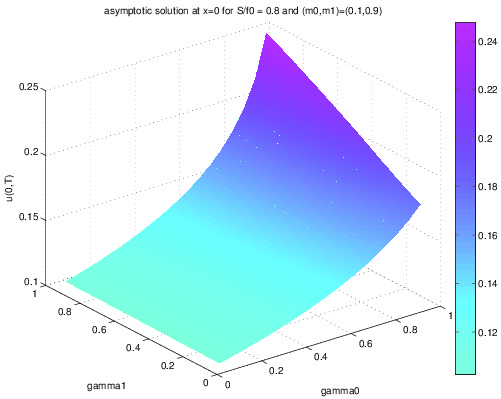}

\\
\includegraphics[scale=0.4]{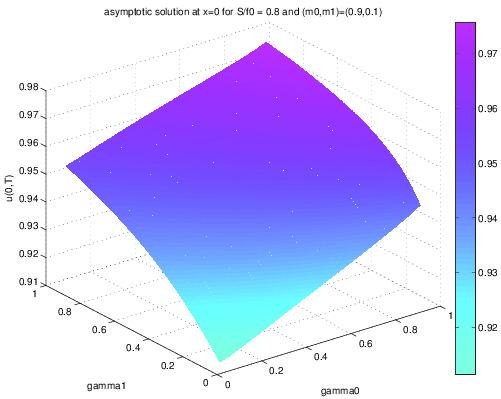}
&
\includegraphics[scale=0.4]{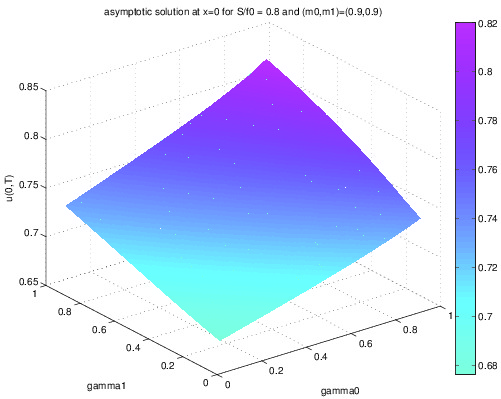}

\end{tabular}
\end{center}
\caption{Coexistence region $C_1$ (mutation  allowed in both directions). Plot of the asymptotic value $\bar u (\gamma_0,\gamma_1)$ in the plane $(\gamma_0,\gamma_1)$ for $\s/f_0=0.3$ fixed and, from top-left to bottom-right, 
$(m_0,m_1)=(0.1,0.1)$, $(m_0,m_1)=(0.1,0.9)$, $(m_0,m_1)=(0.9,0.1)$ and  $(m_0,m_1)=(0.9,0.9)$. These numerical tests show that the asymptotic equilibrium $\bar u(\gamma_0,\gamma_1)$ is a continuous function with respect $\gamma_0$ and $\gamma_1$.} 
\label{fig:2qs_test_1}
\end{figure}


\subsection{Density dependent fitness: the Prisoner's Dilemma.}\label{sec:pd}

In this section we investigate numerically a case example of density-dependent model. 
As noticed before, Propositions \ref{teogamma} and \ref{largetwo} do not apply to this case. On the contrary, the extension of the numerical scheme proposed in the previous section to the case of non-constant fitness, can handle also this more general situation without further difficulties.

We focus on Prisoner's Dilemma, because of the huge interest it receives in the evolutionary dynamics community.
In particular we take the following payoff matrix
\[ A = \left(\begin{array}{cc}  2 &  4\\ 1 & 3\end{array}\right). 
\]
Here, we have  used the convention to  label $0$ defectors and $1$ cooperators, so that the related replicator equation \eqref{r1} has  a stable equilibrium  at $ x=0$.
The replicator-mutator model reads
\begin{align}\label{qsPD} 
\left\{\begin{array}{l}
 \partial_tv = \left((1-2m_0-2m_1)x^2-(1+m_1)x+2m_0\right) \partial_x v,  \\[.15cm]
v(x,0)=x ,\end{array}\right.
\end{align}
and its solution can be computed explicitly. 
Also here the set of the parameters $(m_0,m_1)$ can be split according to the position of the asymptotic equilibrium.
\begin{itemize}
\item[$E$:] the extinction region is $\{0\}\times[0,1]$. Here $\bar x=0$.
\item[$F$:] the fixation region is $[1/4,1]\times\{0\}$. Here $\bar x=1$.
\item[$C_1$:] the global coexistence region is the internal region $(0,1]^2$. Here there is  a globally stable equilibrium at $\bar x \in(0,1)$.
\item[$C_0$:] the local  coexistence region is $=(0,1/4)\times\{0\}$. Here $x=1$ is an unstable equilibrium, and there is a stable equilibrium at $\bar x \in(0,1)$, whose basin of attraction is $[0,1)$.
\end{itemize}

On the other hand, the point-process mutation equation reads 
\begin{align}\label{KPD}
\left\{\begin{array}{rl}
  \partial_tu_{\gamma_0,\gamma_1} = &\!\!\!\! - x (1\!-\!x) \partial_xu_{\gamma_0,\gamma_1} - m_1 (1\!+\!2x) x\, \mathcal{J}(u_{\gamma_0,\gamma_1},-\gamma_1 x)
\\[.15cm]
&\!\!  + m_0 (2\!+\!2x)(1\!-\!x) \, \mathcal{J} (u_{\gamma_0,\gamma_1},\gamma_0(1\!-\!x))  , 
\\[.15cm]
u_{\gamma_0,\gamma_1}(x,0)=&\!\!\!\! x. \end{array}\right.
\end{align}
Remark \ref{constss} fits with this case, so that we can foresee that
\begin{enumerate}[i) ]
\item if $m_0+m_1<1/2$, then $u_{\gamma_0,\gamma_1}\ge v$ pointwise, for any value of the $\gamma_i's$,
\item  if $m_0+m_1>1/2$, then $u_{\gamma_0,\gamma_1}\le v$ pointwise, for any value of the $\gamma_i's$,
\item if $m_0+m_1=1/2$, then replicator-mutator and point-mutation models coincide, actually  $u_{\gamma_0,\gamma_1}= v$ for any value of the $\gamma_i's$.
\end{enumerate}
This means that the time concentration of mutations favours the low-fitness type, as in the quasispecies framework, only below a critical value of the global mutation probability. 
Numerical simulations comply with this relations. 
\\
Figure \ref{fig:prisoner_shorttime_all} illustrates the short-time behaviour of the point-process mutation model, for different values of the mutation parameters $(m_0,m_1)$. In each box  the solutions corresponding to various values of the parameters $\gamma_0,\gamma_1$ are plotted.
Boxes (a) and (b) refer to fair mutation, so $(m_0,m_1)$ is in the extinction region $E$. 
Boxes (c) and (d) refer to unfair mutation: in (c) $(m_0,m_1)$ is in the coexistence region $C_0$, while in (d) $(m_0,m_1)$ is in the fixation region $F$.
Boxes (e) and (f) refer to back and forth mutations, so $(m_0,m_1)$ is in the coexistence range $C_1$. 
In the left side (figures  a, c, e) $m_0+m_1<1/2$: all the solutions stay above the replicator-mutator, and they increase with the $\gamma_i$'s (as for the quasispecies).
In  the right side (figures b, d, f) $m_0+m_1>1/2$ and the order is reversed.

\begin{figure}[htp]
\begin{center}
\begin{tabular}{cc}
\includegraphics[scale=0.4]{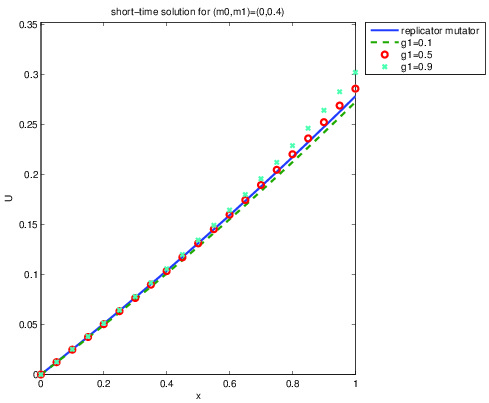}
&
\includegraphics[scale=0.4]{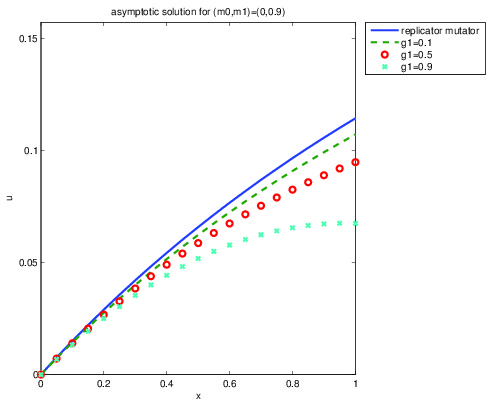}
\\
(a) & (b)
\\
\includegraphics[scale=0.4]{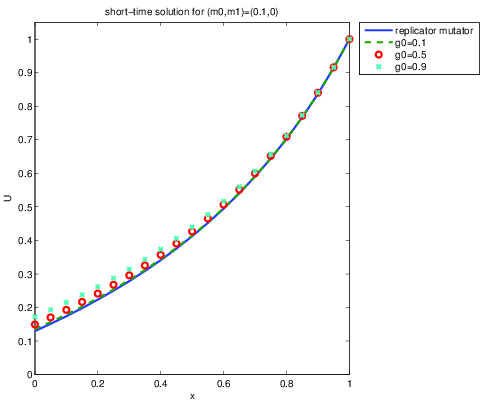}
&
\includegraphics[scale=0.4]{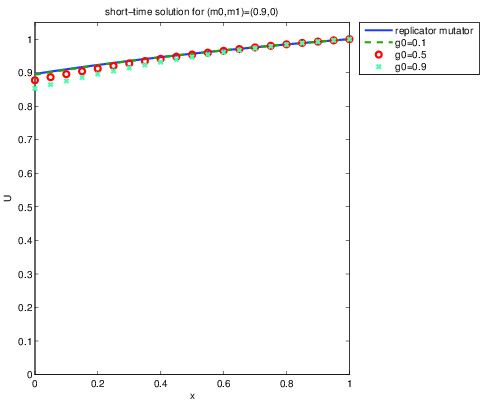}
\\
(c) & (d)
\\
\includegraphics[scale=0.4]{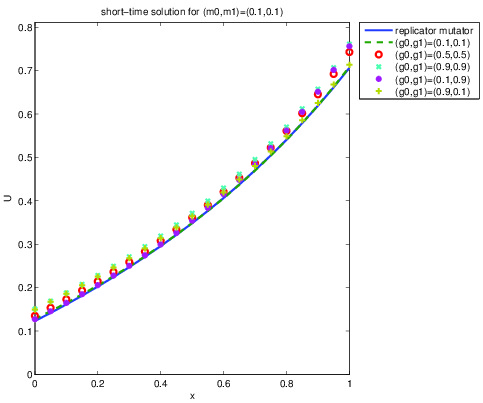}
&
\includegraphics[scale=0.4]{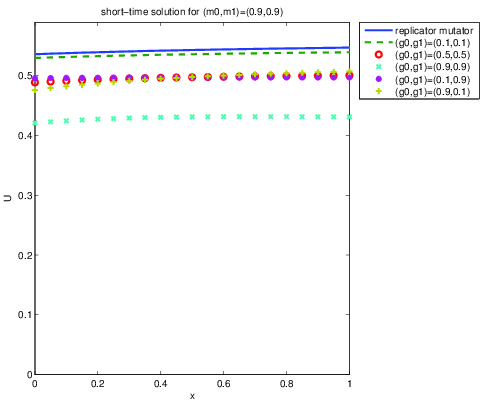}
\\
(e) & (f)
\end{tabular}
\end{center}
\caption{\textbf{Prisoner's Dilemma: short time behaviour}. Plot of the solution $u_{\gamma_0,\gamma_1}(x,T)$ for a fixed short time $T$ and several values of $\gamma_0,\gamma_1$, varying the value of $m_0,m_1$.
For $m_0+m_1<1/2$, the graph of $u_{\gamma_0,\gamma_1}$ stays above the replicator-mutator's one and increases with $\gamma_0, \gamma_1$ (figures (a), (c), (e)).   The situation is reversed for $m_0+m_1>1/2$ (see figures (b), (d), (f)).
} 
\label{fig:prisoner_shorttime_all}
\end{figure}

Figure \ref{fig:prisoner_longtime_all} reports the numerical simulations of long-time asymptotics, that can not be predicted by theoretical tools, so far. 
There are many similarities with the quasispecies case, for instance there is a constant asymptotic equilibrium, which is global if $(m_0,m_1) \in E\cup F\cup C_1$, while in $C_0$ there still is an instable equilibrium at $x=1$.
When $(m_0,m_1)$ is in the extinction region $E$, the point-process model gives extinction for all values of $\gamma_0, \gamma_1$ (see figures (a) and (b)). Also the fixation range is confirmed, as showed in (d).
A new feature shows up in the coexistence range $C_1\cup C_0$. Actually for $m_0+m_1<1/2$, the asymptotic equilibrium stays above the replicator-mutator's one and increases with $\gamma_0, \gamma_1$ (figures  (c), (e)). 
The opposite happens when $m_0+m_1>1/2$:  Figure \ref{fig:prisoner_longtime_all}-(d) reveals that  the time concentration of mutations can also favourite the high-fitness specie, since  the equilibrium stays below the replicator-mutator's one and decreases with $\gamma_0, \gamma_1$.
\\
Figure \ref{fig:prisoner_surf_1} plots the equilibrium $\bar u$ in the plane $(\gamma_0,\gamma_1)$, for $(m_0,m_1)$ fixed in the coexistence range.
Seemingly $\bar u$ depends continuously by $\gamma_0,\gamma_1$.
In (a) $m_0+m_1<1/2$: the picture is the same as for the quasispecies, i.e.~$\bar u$ is monotone increasing, separately on $\gamma_0$ and $\gamma_1$ with
\begin{align*} 
\lim\limits_{(\gamma_0,\gamma_1)\to(0,0)}\bar u(\gamma_0,\gamma_1)=\bar x , \qquad \lim\limits_{(\gamma_0,\gamma_1)\to(1,1)}\bar u(\gamma_0,\gamma_1) = \bar u(1,1) = \frac{2m_0}{2m_0+3m_1} > \bar x. 
\end{align*}
On the contrary, in  (b)  $m_0+m_1>1/2$, then $\bar u$ is monotone decreasing and
\begin{align*}
\lim\limits_{(\gamma_0,\gamma_1)\to(0,0)}\bar u(\gamma_0,\gamma_1)=\bar x , \qquad \lim\limits_{(\gamma_0,\gamma_1)\to(1,1)}\bar u(\gamma_0,\gamma_1) = \bar u(1,1) = \frac{2m_0}{2m_0+3m_1} < \bar x.
\end{align*}

\begin{figure}[htp]
\begin{center}
\begin{tabular}{cc}
\includegraphics[scale=0.4]{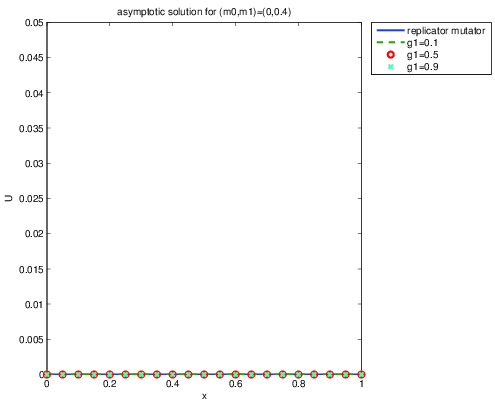}
&
\includegraphics[scale=0.4]{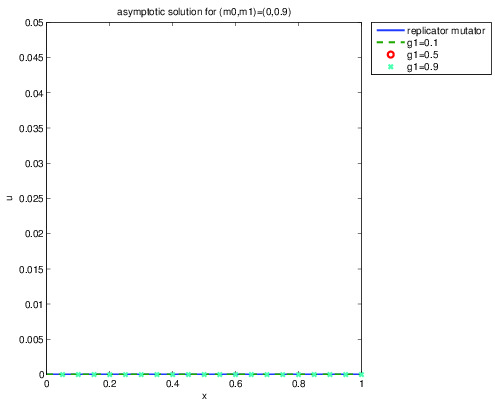}
\\
(a) & (b)
\\
\includegraphics[scale=0.4]{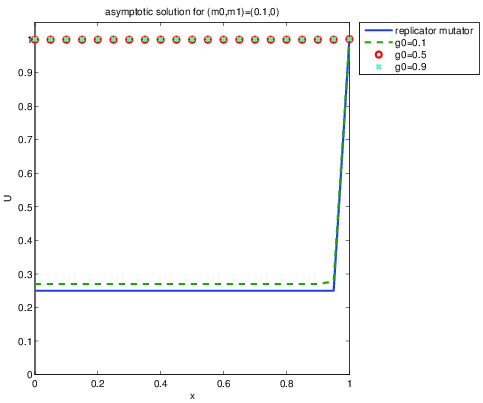}
&
\includegraphics[scale=0.4]{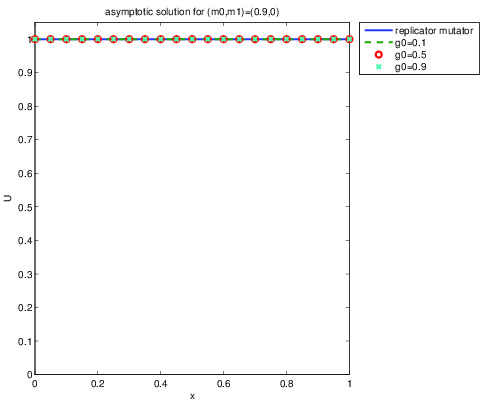}
\\
(c) & (d)
\\
\includegraphics[scale=0.4]{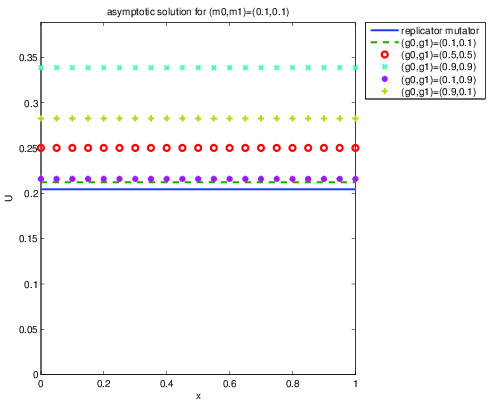}
&
\includegraphics[scale=0.4]{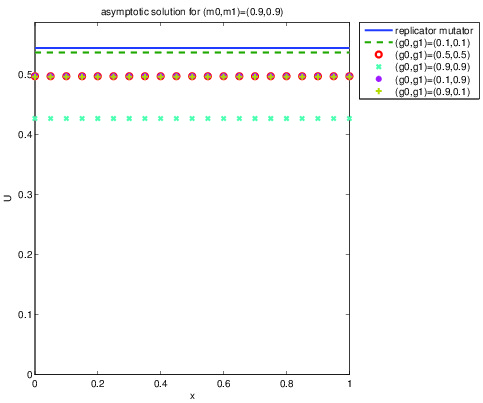}
\\
(e) & (f)
\end{tabular}
\end{center}
\caption{\textbf{Prisoner's Dilemma: long time behaviour}. Plot of the asymptotic solution $\bar u_{\gamma_0,\gamma_1}(x)$ for  several values of $\gamma_0,\gamma_1$, varying the value of $m_0,m_1$ as in Figure \ref{fig:prisoner_shorttime_all}.
In any case the asymptotic equilibrium is constant w.r.t.~$x$, and the extinction/fixation/cohexistence ranges are preserved. 
In the coexistence region, the asymptotic equilibrium depends by $\gamma_0, \gamma_1$. For $m_0+m_1<1/2$, it stays above the replicator-mutator's one and increases with $\gamma_0, \gamma_1$ (figures  (c), (e)). The picture is reversed for $m_0+m_1>1/2$ (figure  (d)).
} 
\label{fig:prisoner_longtime_all}
\end{figure}

\begin{figure}[htp]
\begin{center}
\begin{tabular}{cc}
\includegraphics[scale=0.4]{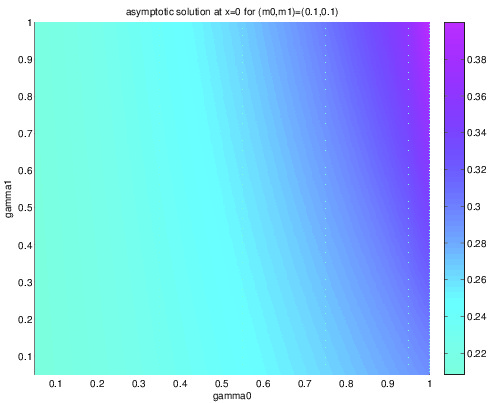}
&
\includegraphics[scale=0.4]{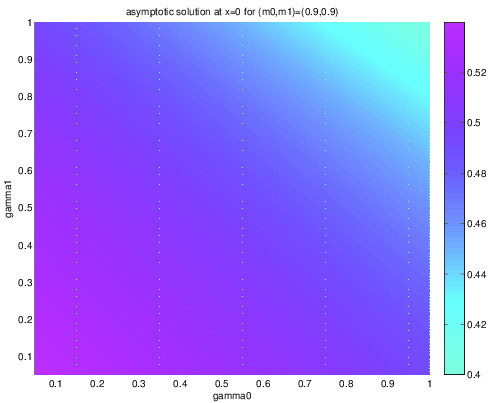}
\\
(a) & (b)
\end{tabular}
\end{center}
\caption{\textbf{Prisoner's Dilemma}. Plot of the asymptotic equilibrium $\bar u$ in the plane $(\gamma_0,\gamma_1)$. In Figure (a)  we set $(m_0,m_1)=(0.1,0.1)$ such that $m_0+m_1 < 1/2$, and in Figure (b) we set $(m_0,m_1)=(0.9,0.9)$ such that $m_0+m_1 > 1/2$.
As noticed before, the order is reversed. Besides, the asymptotic equilibrium seems to be a continuous function of $\gamma_0, \gamma_1$.
} 
\label{fig:prisoner_surf_1}
\end{figure}

\section{Conclusions}\label{4}
We have analyzed an integro-differential model for the evolution of populations with point-type mutations proposed in \cite{ACNT}. 
Some extensions of the analytical qualitative properties established before has been presented. 
It was already known that the modified quasispecies model always favours the low-fitness specie, if compared to the standard selection-mutation model. This is not true, in general, in the  density dependent case, where the high-fitness species can also increase its selection advantage during invasion. If this happens, and if the gain is large with respect to the total amount of mutations (both fair and unfair), then the high fitness specie is  benefited  by the point-type mutations. 
\\
Next, a numerical scheme has been proposed, and it has been shown that it can efficiently handle the blow-up of the spatial derivative of the solution, even for large times. Thanks to our numerical scheme,  we  investigated some questions that were left open by a purely analytical approach. In particular, we found that the asymptotic equilibrium of the solutions  depends continuously on the parameters. Moreover, we observed that our scheme works also for density dependent fitness. The relevant example of Prisoner's Dilemma was taken as a case study: it was shown that there is still  a constant asymptotic equilibrium, that depends continuously and monotonically on the parameters. We have also given an explicit example in which the asymptotic  equilibrium presents respectively a higher or a lower concentration of high-fitness type (w.r.t.~the standard replicator-mutator model), depending on the time intensity of the mutation process.

\end{document}